\documentclass[a4paper,11pt]{article}
\pdfoutput=1
\usepackage{jheppub}

\usepackage{graphicx}
\usepackage{epstopdf}
\newcommand{\be}{\begin{equation}}
\newcommand{\ee}{\end{equation}}
\usepackage{amsmath}
\usepackage{amsfonts}
\usepackage{bm,bbm}
\usepackage{color}
\usepackage{subfigure}
\usepackage{amsthm}
\usepackage[latin1]{inputenc}
\usepackage{tikz}
\usepackage{comment}

\newcommand{\bra}[1]{\left\langle #1 \right|}
\newcommand{\ket}[1]{\left|#1\right\rangle}
\newcommand{\braket}[2]{\left\langle#1 |  #2\right\rangle}

\newcommand{\1}{\mathbbm{1}}

\newcommand{\cH}{\mathcal{H}}
\newcommand{\cC}{\mathcal{C}}
\newcommand{\cR}{\mathcal{R}}
\newcommand{\cB}{\mathcal{B}}
\newcommand{\tr}{{\rm tr}}

\newtheorem{lem}{Lemma}
\newtheorem{thm}{Theorem}
\newtheorem{defi}{Definition}
\newtheorem{corollary}{Corollary}

\title{Entanglement renormalization, quantum error correction, and bulk causality}
\author[a]{Isaac H. Kim}
\affiliation[a]{IBM T. J. Watson Research Center, Yorktown Heights, New York, USA}
\author[b]{Michael J. Kastoryano}
\affiliation[b]{NBIA, Niels Bohr Institute, University of Copenhagen, Denmark}
\emailAdd{ikim@us.ibm.com}
\emailAdd{kastorya@nbi.ku.dk}
\date{\today}

\abstract{
	Entanglement renormalization can be viewed as an encoding circuit for a family of approximate quantum error correcting codes.
	The logical information becomes progressively more well-protected against erasure errors at larger length scales.
	In particular, an approximate variant of holographic quantum error correcting code emerges at low energy for critical systems. This implies that
	two operators that are largely separated in scales behave as if they are spatially separated operators, in the sense that they obey a Lieb-Robinson type
	locality bound under a time evolution generated by a local Hamiltonian.}
\begin{document}
      \maketitle
\flushbottom

\section{Introduction}
In most physical theories, the notion of locality is imposed, as opposed to being derived from more elementary principles. The AdS/CFT correspondence
indicates that this picture may need to be amended, at least for studying the quantum theory of gravity \cite{Maldacena1997,Witten1998,Aharony1999}. An interpretation
of the duality in the language of the quantum error correcting codes\cite{Almheiri2014}, and the proposal that spacetime may be built out of entanglement \cite{Raamsdonk2010}, suggests
a fruitful avenue along which we can study these questions in the language of quantum information theory.

There has been a recent surge of activity devoted to constructing holographic quantum error correcting codes\cite{Pastawski2015,Yang2015,Hayden2016,Donnelly2016}. These are families of codes which can be formally
expressed as an encoding map from the bulk theory to the boundary theory or vice versa. While the details behind these codes vary, they share
 a number of interesting properties. The operators in the bulk can be mapped to operators on the boundary which obeys certain quantum error correction
 properties outlined in Ref.\cite{Almheiri2014}. They can also reproduce, to some extent, the celebrated Ryu-Takayanagi formula \cite{Ryu2006}.

However, several important issues remain unresolved. Most importantly, these codes are constructed from scratch, as opposed to being derived from a set of well-motivated
assumptions. If we believe in the unitary equivalence of CFT and the quantum theory of gravity in AdS space, we should be able to explain how such codes
emerge from the properties of the CFT. Second, the question of dynamics remains open. Modulo one exception \cite{Yang2015}, these codes are formally maps from the bulk to the boundary that
are injective but not surjective. Therefore, acting a Hamiltonian on the code state will generically produce a state that is outside the code subspace. Furthermore, the boundary Hamiltonian, even if it is local, becomes generically non-local once it is mapped to an operator in the bulk. Resolving
what it means to have  causal bulk dynamics in the presence of these complications is clearly a nontrivial problem.

The purpose of this paper is to make progress on these important issues.
First, we show that an approximate version of holographic quantum error correcting code emerges
at low energy at criticality at scales large compared to the AdS radius, if the ground state can be well-approximated by a certain multi-scale entanglement renormalization ansatz (MERA) \cite{Vidal2007} for which correlations
decay polynomially with distance. Empirical evidences suggest that this is likely to be true for quantum spin systems at criticality \cite{Pfeifer2009}. If that is indeed
true, our work implies that certain variants of holographic quantum error correcting codes naturally emerge in these systems. We also derive fundamental bounds on the error correcting capabilities of these codes.
As for the dynamics, we derive a Lieb-Robinson
type locality bound \cite{Lieb1972} between two observables that are largely separated in scale. It is important to note that these observables generally do not even
commute with each other. Despite this fact, they behave as if they were spatially separated operators undergoing a dynamics generated by a
local Hamiltonian. In some sense, the causal dynamics in the bulk emerges from the universal structure of entanglement at low energy.

Our work supports the proposals to interpret MERA as a discrete analogue of the AdS/CFT correspondence\cite{Swingle2012,Swingle2012a,Qi2013,Miyaji2015}, at least at scales large compared to the AdS radius.
In order to be able to accommodate locality at sub-AdS scale, one would need to incorporate more fine-grained structures. In its present form, our conclusion is so general
that it is even applicable to free-fermion systems, which is unlikely to admit a semiclassical gravitational dual \cite{Castro2012}.

It has been known that
tensor networks such as MERA lead to constructions of various quantum error correcting codes \cite{Ferris2014}. What is interesting is that, as suggested by Pastawski et al. \cite{Pastawski2016},
these codes naturally appear at low energies of critical systems. These codes differ greatly from the so called topological codes\cite{Kitaev1997} in that (i) erasures of bounded regions can be
corrected up to a polynomially small error, rather than exponentially small error and that (ii) one in fact has a family of codes that are related to the geometric data of the
hyperbolic space. Our work provides a concrete framework and technical tools from which the structure of these codes can be studied.

The results presented here rely on a very general property of entanglement renormalization and on recent insights from the theory of approximate quantum
error correction (AQEC) \cite{Flammia2016}. It is particularly illuminating to use entanglement renormalization in the ``Heisenberg picture," wherein the renormalization group (RG)
flow acts on the space of observables. This is an observation already made in the literature \cite{Evenbly2007,Giovannetti2008}, which we generalized substantially in this paper.
The only property that we use is the fact that this RG flow (i) preserves locality and that (ii) it is norm-nonincreasing.
Both of these properties are manifestly true for various proposed forms of entanglement renormalization, but they are not the only possibilities. While we restrict
ourselves to one-dimensional systems for concreteness, it should be clear that these two properties are sufficient to guarantee a similar conclusion in more
generalized settings, e.g., higher dimensions and different spacetime geometry. The insight that we bring from AQEC is a duality relation between decoupling and recoverability; the degree to which  quantum
information can be recovered from a given region is exactly equal to the degree to which certain regions are completely decoupled from each other\cite{Kretschmann2008,Beny2010,Flammia2016}.

In Section \ref{section:entanglement_renormalization}, we review basic facts about entanglement renormalization and derive several identities that form the basis of our analysis.
In Section \ref{section:qec_holography}, we sketch the relation between entanglement renormalization and error correction in the context of holography. We then derive fundamental
properties of error correcting codes that emerge from entanglement renormalization in Section \ref{section:mera_qec}. In Section \ref{section:applications}, we use these properties to constrain
the support of the logical operator and derive a Lieb-Robinson type locality bound between two bulk observables.

\section{Entanglement Renormalization\label{section:entanglement_renormalization}}
Entanglement renormalization was introduced by Vidal \cite{Vidal2007,Vidal2008} to numerically study the critical behavior of one-dimensional quantum many-body
systems. Generalizations to higher dimensions are known \cite{Evenbly2009}. We review the basic notions underlying these constructions, and review several facts that
are pertinent to this paper. MERA is a many-body quantum state that is created by applying a  quantum circuit to a simple product state, say $\ket{0}^{\otimes N}$,
where $N$ is the number of qubits.

There are two important properties that underlie this circuit, and these will form the basis of our argument. First, the circuit is hierarchical. It can be decomposed
into a sequence of isometries, which are labeled in terms of a parameter ($s$) that ranges from $0$ to $O(\log N)$. These isometries will play an important role; we
denote them  $W_s$. It should be noted that the isometry $W_s$ maps vectors of the Hilbert space at ``scale" $s$ to
the Hilbert space at scale $s-1$. These Hilbert spaces are denoted  $\mathcal{H}_s$. In particular, $\mathcal{H}_0$ is the physical Hilbert space. Second, at every level $s$, the isometry $W_s$ preserves locality.  Applying the dual of these isometries to a local operator results in another local operator; that is,
the support of $W_s O_s W_s^{\dagger}$ can only be larger than the support of $O_s$ by a constant amount, where $O_s$ is acting on $\cH_s$.

In the original construction of Ref. \cite{Vidal2007}, $W_s$ is the composition of a global product of isometries and disentanglers at scale $s$: $W_s:= \otimes_{x_s} V_{x_s}\otimes_{y_s}U_{y_s}$, where $x_s$ and $y_s$ are an index of the position along the chain at level $s$. In Figures \ref{fig:MERA1} and \ref{fig:MERA2}, for convenience, we will illustrate binary MERA constructions with uniform isometries $V_{x_s}=V$ and unitaries $U_{x_s}=U$,  however our results hold for the more general construction above.

While MERA is usually a single state, we will instead consider a family of subspaces, $\cC_s$. These subspaces are defined in terms of the isometries from $\mathcal{H}_s$ to $\mathcal{H}_0$: $\cC_s=\{ W_1 \cdots W_s\ket{\varphi_s} | \ket{\varphi_s}\in\cH_s\}$.
In Fig. \ref{fig:MERA1}, we have illustrated $\cC_5$ for $W_s:= \otimes_{x_s} V_{x_s}\otimes_{y_s}U_{y_s}$. It should be clear that for any finite MERA construction there exists an $s_{\rm max}=O(\log(N))$ such that $\cC_s$ is trivial for all $s\geq {s_{\rm max}}$.

\begin{figure*}
	\begin{center}
		\includegraphics[scale=0.30]{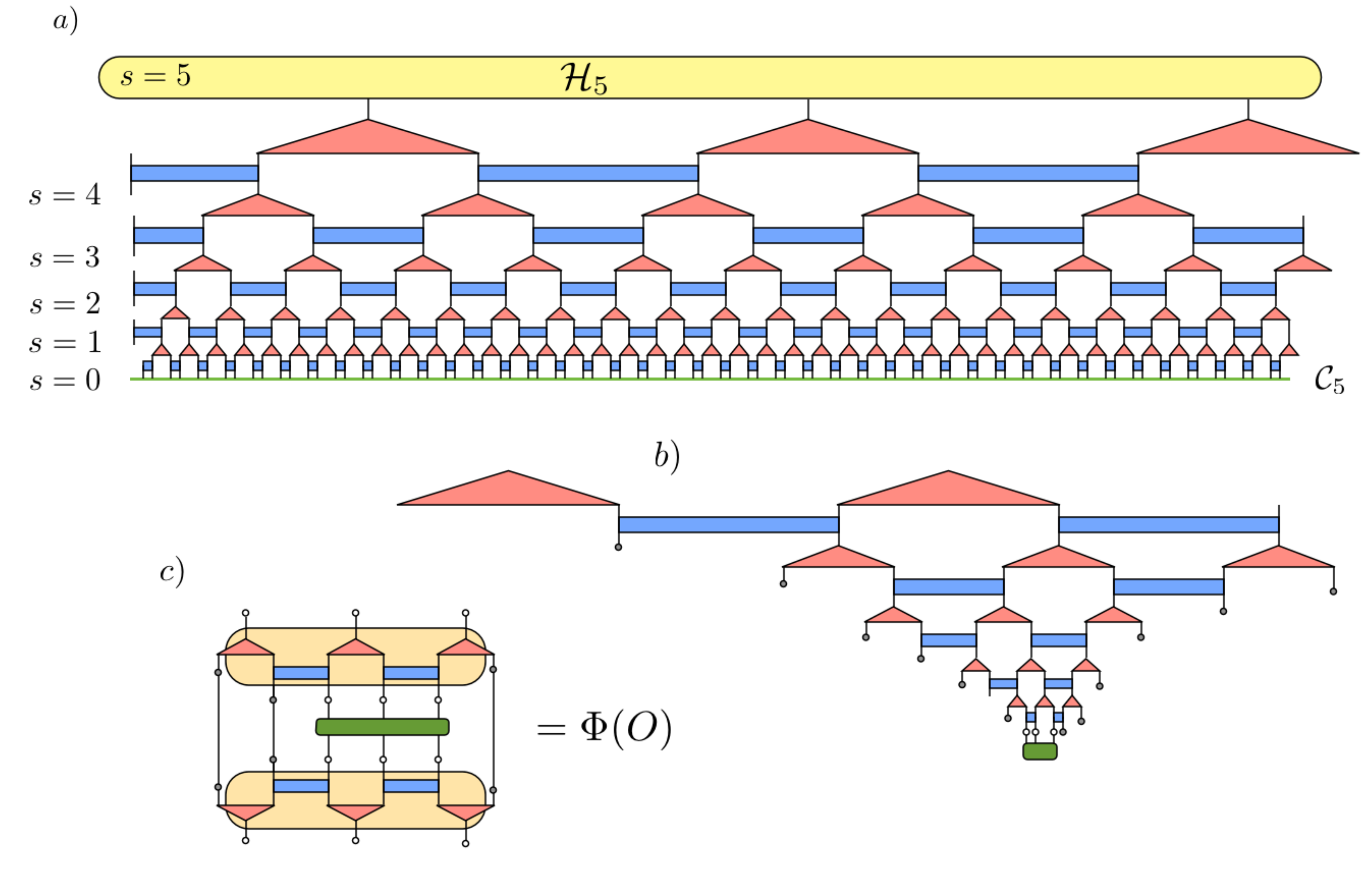}
	\caption{
       We illustrate a scale invariant MERA construction. The blue rectangles are the disentangling unitaries $\{U_{x_s}\}$, and the red triangles are the isometries $\{V_{x_s}\}$. The green blocks represent local observables. a) A MERA network up to level $s=5$ with the isometries $W_s:= \otimes_{x_s} V_{x_s}\otimes_{y_s}U_{y_s}$. The yellow box at the top of the figure represents the Hilbert space $\cH_5$ and can be considered as the bare logical vectors. The MERA circuit serves as an encoding map from the logical space $\cH_s$ onto the physical one $\cC_s$. b) The past causal cone of a local observable in the physical space (the boundary). Locality (3-adjacent) of the observable is preserved at all levels of the network. c) The transfer operator $\Phi(\cdot)$ acting on an elementary block.
		}\label{fig:MERA1}
			\end{center}
\end{figure*}

For further analysis, it will be convenient to work in terms of a certain family of purified states. Consider a state $\rho$ acting on $\mathcal{H}_s$. We would like
to consider a family of states that are (i) first purified and (ii) then mapped into the Hilbert space $\mathcal{H}_{s'}$($s'<s$) by applying an isometry $W_{s'}W_{s'+1}\cdots W_s$.
In concrete terms, such a state is expressed as follows:
\begin{equation}
	\ket{\rho_{s'}}:=W_{s'}W_{s'+1}\cdots W_s(\rho_s^{1/2}\otimes U_{R_s})\ket{\Omega_s},
\end{equation}
where $\ket{\Omega_s}$ is a maximally entangled state between $\mathcal{H}_s$ and a copy of $\cH_s$ which we call $\cH_{R_s}$, and $U_{R_s}$ is a unitary operator acting on $\cH_{R_s}$. In particular, $\ket{\rho_s}$
is  a purification of $\cC_s$: $\tr_{R_s}[\ket{\rho_s}\bra{\rho_s}]=\rho_s$.


\subsection{Renormalization in the Heisenberg Picture}
The MERA formalism is especially well suited to studying expectation values of local observables. More generally, we will need to  consider objects of the form:
\begin{equation}
	\bra{\rho_s} O_s \ket{\sigma_s},\label{eq:mera_object}
\end{equation}
where $O_s$ is some operator that is supported on $\cH_s\otimes\cH_{R_s}$ and $\ket{\rho_s}, \ket{\sigma_s}$ are purifications of $\rho_s,\sigma_s\in \cC_s$, and  $O_s$ will often have some additional locality structure on $\cH_s$. The reason for considering such objects will become evident once we explain its relation to
quantum error correction in Section III. For the moment though, it will be important to develop the machinery for their analysis.

For that purpose, it will be convenient to recast this object in an alternative form, which can be thought as the ``Heisenberg picture" of entanglement renormalization.
Let us first note the following identity:
\begin{equation}
	\bra{\rho_s} O_s \ket{\sigma_s} = \bra{\rho_{s+1}} \Phi_s^{s+1}(O_s)\ket{\sigma_{s+1}},
\end{equation}
where $\Phi_s^{s+1}(\cdot) = W_{s+1}(\cdot) W_{s+1}^{\dagger}$. This map is completely-positive, trace-preserving (CPTP) and unital. Such maps are often referred to as (unital) quantum channels.  In particular, it is norm-nonincreasing
and maps the identity operator to the identity operator. We shall refer to the process of applying $\Phi_s^{s+1}$ to $O_s$ as the process of coarse-graining (renormalizing)  the operator from scale
$s$ to $s+1$. More generally, we will consider the map:
\begin{equation}
	\Phi_s^{s'} := \Phi_{s'-1}^{s'}\circ \cdots \circ \Phi_s^{s+1},
\end{equation}
which corresponds to the process of renormalizing an operator from scale $s$ to $s'$, where $s'>s$. It is clear that $\Phi_s^{s'}$ maps operators on $\cH_{s}$ to operators on $\cH_{s'}$.

Under the renormalization procedure, the evolution of the operator can be broken down into two stages. In the first stage, the support of the operator shrinks monotonously, at a constant rate: if $A_s$ is a simply connected region at level $s$, then an operator $O_{A_s}$ supported on $A_s$ gets mapped to an operator $O_{A_{s+1}}\equiv \Phi_s^{s+1}(O_{A_s})$, where $|A_{s+1}|\leq c|A_s|$ for some constant $c>1$. In the (binary) MERA network illustrated Fig. \ref{fig:MERA1}, the constant $c$ is $2$.
The set of the supports over different scales, $\{A_s,A_{s+1},\cdots, A_{s'-1},A_{s'} \}$, is said to be the \emph{past causal cone of $A$} from $s$ to $s'$. When the range is obvious from the
context, we shall simply say past causal cone, without specifying the range.

In other words, as an operator is renormalized from one scale to another, its support $(\ell)$ shrinks exponentially with the scale separation. After $O(\log \ell)$ renormalization steps,  the support size becomes $O(1)$, and the second stage begins. What distinguishes the second stage from the first is the fact that the support of the operator remains
 constant under further coarse graining. The minimal nontrivial regions which can support such operators are referred to as the \emph{elementary blocks} of the MERA network (see Fig. \ref{fig:MERA1}(b-c)).


The aforementioned behavior of renormalized operators is, qualitatively speaking, independent of the details of the MERA network. That is, the conclusion remains intact even if the shape of the
network differs at different scale or even if there is a spatial anisotropy. However, accommodating those generalizations will necessitate unnecessary complications. This is why we shall consider
MERA networks that are  \emph{scale-invariant}, which we define below.

\begin{defi}
	A MERA network is scale invariant if there exists an isometry $V$ and a unitary $U$ such that $V_{x_s} = V$ $\forall x_s,s$ and $U_{y_s}=U$ $\forall y_s, s$.
\end{defi}

$\Phi_{s}^{s'}$ is a quantum channel that maps operators on $\cH_s$ to operators on $\cH_{s'}$, which will typically be different spaces. However, if the MERA network is scale invariant, when $\Phi_s^{s'}$ acts on an observable in an elementary block of $s$, it gets mapped to an observable in an elementary block in $s+1$. The unique channel $\Phi$ mapping operators between elementary blocks of $s$ and $s+1$ can be represented as one with identical input and output space (see Fig. \ref{fig:MERA1}a). This is extremely convenient as it allows us to map a ``trail'' of elementary blocks up the MERA network as the iteration of quantum channels (Fig.\ref{fig:MERA1}b). If the network is scale invariant, as is expected for critical systems, the dynamics between elementary blocks is governed by stationarity and mixing properties of the channel $\Phi$, which will also be referred to as the \textit{transfer operator}.

Generic quantum channels (see the appendix for a discussion) that have the same input and output algebra can be written as
\begin{equation}\label{eqn:channel}
\Phi(O)=\sum_k \lambda_k \tr[OR_k]L_k,
\end{equation}
where $\lambda_k$ are the eigenvalues of $\Phi$, and $L_k,R_k$ are the bi-orthonormal left and right eigenvectors: $\tr[L_k,R_l]=\delta_{kl}$. The spectrum of the channel is bounded by one ($|\lambda_k|\leq 1$), and for generic quantum channels, there is only one eigenvalue of magnitude $1$ corresponding to the unique stationary state of the channel (in the Schr\"odinger picture). Arranging the eigenvalues in decreasing order (decreasing real part), we get that $\lambda_0=1$, with  $L_0=\1$ and $R_k=\rho_{\rm ss}$ is a density matrix, which we will refer to as the stationary state for the elementary block. $\lambda_1$  will play an important role in the remainder of the paper. For scale invariant MERA, $\nu:=-\log_2({\rm Re}\lambda_1)$ will be referred to as the \textit{scaling dimension}.

To conclude this section, we formally define the class of channels that we plan to work with:
\begin{defi}
For a scale invariant MERA network defined by the isometries $\{W_s:= \otimes_{x_s} V_{x_s}\otimes_{y_s}U_{y_s}\}$, we say that the class of channels $\Phi_s^{s+1}(\cdot)=W^\dag_s(\cdot)W_s$  is {\rm RG-regular} if its action on elementary blocks can be written as in Eq. (\ref{eqn:channel}) with a scaling dimension $\nu:=-\log_2({\rm Re}[\lambda_1])$ strictly larger than zero.
\end{defi}

If the subspaces $\cC_s$ are related to $\cH_s$ by an RG-regular channel $\Phi_0^s$, we will say that $\cC_s$ are RG-regular subspaces (or \textit{codes} in the error correction language).

\subsection{Calculus for entanglement renormalization}
The action of the renormalization map $\Phi_s^{s'}$ on general, \emph{non-local} operators play an important role. We develop a calculus that facilitates this analysis below. The operators that we consider are, generally speaking, supported on three subsystems, which we denote as $A,A', $and $R$. Here $A$ is a subsystem of the physical Hilbert space($\mathcal{H}_0$), $R$ is the purifying space, and $A'$ is yet another subsystem that is included neither in the physical Hilbert space nor in the purifying space. Let us denote such an operator as $O_{AA'R}$. Simply connected regions of the physical Hilbert space $\cH_0$ will be denoted without a subscript ($A$).

We consider a linear map of the following form:
\begin{equation}
O_{AA'R} \to \bra{\rho_0} O_{AA'R} \ket{\rho_0},
\end{equation}
It is important to note that the output of this map is generally an operator, because $A'$ lies outside of the physical Hilbert space and the purifying space. We see that
\begin{align}
\bra{\rho_0} O_{AA'R} \ket{\rho_0} &= \bra{\rho_s} \Phi_{0}^{s}(O_{AA'R})\ket{\rho_s} \nonumber \\
&= \tr_{A_sR}(\rho_s^{A_sR} \Phi_0^s(O_{AA'R})).\label{eq:renormalization_identity}
 \end{align}
The first line follows from the definition of the state. The second line follows from the locality of the renormalization map; it maps an operator supported on $AA'R$ to an operator supported on $A_sA'R$. We will often write subsystems as superscripts in order to specify the MERA scale $s$ in the subscript.

While this identity in Eq. (\ref{eq:renormalization_identity}) may seem a bit obtuse, it has important implications. First, consider the case in which $A'$ is an empty set. The correlations between $A$ and $R$ for an arbitrary operator have a simple closed-form expression.
\begin{lem}\label{lem:clustering_identity_simple}
For any state $\rho_0\in\cH_0$,
\begin{equation}
\tr[\rho_0^{A} \otimes \rho_0^{R}O_{AR}] = \tr[\rho_s^{A_s}\otimes \rho_s^R\Phi_0^s(O_{AR})]. \label{eq:clustering_identity_simple}
\end{equation}
\end{lem}
\begin{proof}
First note that $\tr_A[\rho_0^A O_{AR}] = \tr_{A_s}[\rho_s^{A_s} \Phi_0^s(O_{AR})]$ by viewing $R$ as the subsystem $A'$ in Eq.\ref{eq:renormalization_identity}.
Therefore,
\begin{equation}
	\tr_R[\rho_0^R\tr_A[\rho_0^AO_{AR}]] = \tr_R[\rho_0^R\tr_{A_s}[\rho_s^{A_s}\Phi_0^s(O_{AR})]]
\end{equation}
\footnote{Correspondingly, in view of Eq.\ref{eq:renormalization_identity}, $O_{AR}$ should not be viewed as an operator supported on $A$ and the purifying space. It should be instead viewed as an operator supported on $A$ and a subsystem $R$ which is neither in the physical Hilbert space nor in the purifying space.} The claim follows from the trivial identity $\rho_0^R=\rho_s^R$.
\end{proof}

There is in fact a more general identity, which plays a crucial role in our analysis.
\begin{lem}\label{lem:clustering_identity_general}
If $A_{s'} \cap C_{s'} =\emptyset $, $\forall s'\leq s$,
\begin{equation}
\tr[\rho_0^{A} \otimes \rho_0^{CR}O_{ACR}] = \tr[\rho_s^{A_s}\otimes \rho_s^{C_sR}\Phi_0^s(O_{ACR})]. \label{eq:clustering_identity_general}
\end{equation}
\end{lem}
\begin{proof}
Consider an operator Schmidt decomposition of $O_{ACR}$:
\begin{equation}
O_{ACR} = \sum_j O_{A,j}\otimes O_{CR,j},
\end{equation}
where $O_{A,j}$ is an operator supported on $A$ and $O_{CR,j}$ is an operator supported on $CR$. Because any operator admits such a decomposition, it suffices to prove the statement for an operator of a tensor product form between $A$ and $CR$. Without loss of generality, consider an operator $O=O_1\otimes O_2$, where $O_1$ is supported on $A$ and $O_2$ is supported on $CR$.

\begin{align}
\tr_A[\rho_0^A O_1 \otimes O_2] &= \tr_A[\rho_0^A O_1] O_2 \nonumber\\
&= \tr_{A_s}[\rho_s^{A_s} \Phi_0^s(O_1)] O_2. \nonumber
\end{align}
Note that $\tr_{A_s}[\rho_s^{A_s}\Phi_0^s(O_1)]O_2$ is an operator supported on $CR$, as the term appearing before $O_2$ is a scalar. By using the fact that
$\tr[\rho_0^{CR} O_2] = \tr[\rho_s^{C_sR}\Phi_0^s(O_2)]$,
\begin{align}
\tr[\rho_0^A \otimes \rho_0^{CR} ] = \tr[\rho_s^{A_s}\otimes \rho_s^{C_sR}O_1'\otimes O_2'],
\end{align}
where $O_1' = \Phi_0^s(O_1)$ and $O_2' = \Phi_0^s(O_2)$. Since we assumed that $A_{s'} \cap C_{s'} =\emptyset $ $\forall s'\leq s$, the past causal cone of $A$ and $C$ never overlap with each other in this
range. Therefore, $O_1' \otimes O_2' = \Phi_0^s(O_1\otimes O_2)$. This completes the proof.
\end{proof}

\section{Quantum Error Correction and Holography\label{section:qec_holography}}
Recently, various quantum error correcting codes were proposed as models of holography \cite{Pastawski2015,Yang2015,Hayden2016,Donnelly2016}. These codes are equipped with a family of logical
operators that are labeled by the coordinates in the bulk. The radial coordinate, which in our setup corresponds to the scale($s$), is particularly interesting in the context of quantum error
correction. The logical information becomes progressively more well-protected against erasures of boundary subsystems as it recedes further into the bulk.

We will show that such codes naturally arise from the MERA construction. Our choice of logical operators follow the choice of bulk local operators defined in Ref. \cite{Qi2013,Miyaji2015}. In
our notation, the logical operators at scale $s$ will have the form of $W_s\cdots W_1 O W_1^{\dagger}\cdots W_s^{\dagger}$, where $O\in \mathcal{B}(\mathcal{H}_0)$. We show that, as in
the existing proposals \cite{Pastawski2015,Yang2015,Hayden2016,Donnelly2016}, these operators are more well-protected as $s$ increases. We also derive several fundamental properties of these codes.

How are these results at all related to the discussion in Section II? The answer lies on a well-known duality relation between two different concepts, which is perhaps one of the most
fundamental insights behind quantum error correction.  Erasure of a certain region is correctable if and
only if the region contains no logical information \cite{Kretschmann2008,Beny2010}. In slightly more technically terms, an erasure is correctable if and only if the region is uncorrelated
with the purifying space for all the codewords. This equivalence relation implies that it suffices to bound the correlations between the purifying space and a subsystem of interest. This is why
we considered objects of the form of Eq.\ref{eq:mera_object} in Section II.

It turns out, however, that much more can be learned about the structure of these codes by introducing a more refined notion of correctability. It is the notion of \emph{local correctability}
which was introduced in Ref. \cite{Flammia2016} and used in the context of holography in Ref. \cite{Pastawski2016a}. As in Refs. \cite{Kretschmann2008,Beny2010}, there is a similar duality relation between local correctability and the degree to which different subsystems are uncorrelated from each other. In words, erasure of a region $A$
is locally correctable from a recovery operation on $AB$ if and only if $ACR$ is decoupled into $A$ and $CR$, where $C$ is the complementary region of $AB$ and $R$ is the purifying space.
It should be clear that this subsumes the less general case of $B$ being an empty set, which corresponds to Refs. \cite{Kretschmann2008,Beny2010}. Specifically, this result is encapsulated in
Theorem \ref{thm:decoupling}

\begin{thm}\label{thm:decoupling}
	\cite{Flammia2016} Consider a code $\mathcal{C}$ whose underlying Hilbert space can be decomposed into a tensor product of $A,B,$ and $C$. Let $R$ be the purifying space of $\mathcal{C}$.
	Then the following two objects are equal:
	\begin{equation}
		\min_{\omega^A} \sup_{\rho^{ABCR}} \mathfrak{B}(\omega^A \otimes \rho^{CR},\rho^{ACR})\label{eq:decoupling}
	\end{equation}
	\begin{equation}
		\inf_{\mathcal{R}_B^{AB}}\sup_{\rho^{ABCR}}\mathfrak{B}(\mathcal{R}_B^{AB}(\rho^{BCR}), \rho^{ABCR}),\label{eq:correctability}
	\end{equation}
	where $\inf$ is over all CPTP maps from $\mathcal{B}(\mathcal{H}_B)$ to $\mathcal{B}(\mathcal{H}_{AB})$ and $\mathfrak{B}(\cdot, \cdot)$ is the Bures distance.
\end{thm}

A few remarks are in order. First, the Bures distance is a distance measure that can be easily related to a more familiar one, the trace norm:
\begin{equation}
	2\mathfrak{B}^2(\rho,\sigma)\leq \|\rho-\sigma \|_1 \leq 2\sqrt{2}\mathfrak{B}(\rho,\sigma).\label{eq:Bures_Trace}
\end{equation}
The trace norm between two quantum states has the operational interpretation that it quantifies the probability with which two states can be distinguished by a global measurement.

Second, if $\rho^{ACR}$ is close to $\rho^{A} \otimes \rho^{CR}$, it implies that erasure of region $A$ can be corrected by some map supported on $AB$. This is because
such a factorization implies that the expression in Eq. (\ref{eq:decoupling}) is small, which subsequently implies that the expression in Eq. (\ref{eq:correctability}) is small. The latter equation, in words, says
that the original state is close to the state that is created by (i) erasing $A$ and then (ii) applying some recovery map on $AB$. The converse direction also works. If
there exists a recovery map on $AB$ that can correct the erasure of $A$, then $\rho^{ACR}$ should be close to the form of $\omega^A \otimes \rho^{CR}$ by Theorem \ref{thm:decoupling}.
Because these two states must be also close to each other over their subsystems, $\omega^A$ should be close to $\rho^A$, establishing the converse direction.

To summarize, by exploiting the basic structure of the MERA network, one can tightly bound correlations between two subsystems. This bound in turn, by using Theorem \ref{thm:decoupling}, implies
that erasure of certain regions are correctable. This establishes how well the logical information at different scales are protected.

\section{MERA as an approximate quantum error correcting code\label{section:mera_qec}}

We have already formally defined the code subspace $\cC_s \subset \mathcal{H}_0$. What remains is to study the properties of the code subspace. What
kind of erasures are correctable? If they are correctable, how well can those errors be reversed? As we shall see, the analysis follows naturally from the framework that we have constructed in Sec. \ref{section:entanglement_renormalization}.
We begin by a simple warm-up exercise, wherein we study the correctability of simply connected regions. We then move on to studying the correctability of more general regions and deriving a
fundamental tradeoff bound. The key technical result is Theorem 2, which establishes the local correctability of these codes.


\subsection{Correctability of simply connected regions}
As a warm-up exercise, we show that erasure of any simply connected region $A$ can be approximately corrected up to a small error if $s\gg\log_2|A|$.
\begin{lem}
	\label{lem:simply_connected}
	Let $\cC_s$ be an RG-regular MERA code. Then for any $O_{AR} \in \mathcal{B}(\mathcal{H}_A \otimes \mathcal{H}_R)$ where $A$ is a simply connected region, and any purified code state
	$\rho^{AA^cR}$,
	\begin{equation}
		|\tr[(\rho_0^{AR} - \rho_0^A \otimes \rho_0^R)O_{AR}]| \leq C \| O_{AR} \| 2^{-\nu(s-\log_2|A|)}.\label{eq:correctability_bound}
	\end{equation}
\end{lem}
\begin{proof}
	First recall the following two identities.
	\begin{equation}
	\tr[\rho_0^{AR} O_{AR}] = \tr[\rho_s^{A_sR} \Phi_0^{s}(O_{AR})],
	\end{equation}
	\begin{equation}
		\tr[\rho_0^A \otimes \rho_0^{R} O_{AR}] = \tr[\rho_s^{A_s}\otimes \rho_0^R \Phi_0^{s}(O_{AR})].
	\end{equation}
	The first identity follows trivially from the definition and the second one follows from Lemma \ref{lem:clustering_identity_simple}.
	Let us denote the left hand side of Eq. (\ref{eq:correctability_bound}) as $\delta$. The two identities  above imply
	\begin{align}
		\delta&=\tr[(\rho_s^{A_sR}-\rho_s^{A_s} \otimes \rho_s^{R}) \Phi_0^s(O_{AR})] \nonumber\\
								    &=\tr[(\rho_s^{A_sR}-\rho_s^{A_s} \otimes \rho_s^{R}) \Phi_{r_A}^s \Phi_0^{r_A}(O_{AR})] \nonumber\\
								    &=\tr[(\rho_s^{A_sR}-\rho_s^{A_s} \otimes \rho_s^{R}) \Phi_{r_A}^s (O_{A_{r_A}R})] \nonumber\\
								    &=\tr[(\rho_s^{A_sR}-\rho_s^{A_s} \otimes \rho_s^{R}) (\sum_{j=1}^{d^2} \Phi_{r_A}^s(O_{A_{r_A},j}) \otimes O_{R,j})],\nonumber
	\end{align}
	where $O_{A_{r_A}R}=\Phi_0^{r_A}(O_{AR})$ and $r_A$ is chosen such that $O_{A_{r_A}R}$ is supported on one of the elementary blocks and the purifying space. This operator can,
	without loss of generality,  be decomposed into the operator Schmidt decomposition $\sum_{j=1}^{d^2}  O_{A_{r_A},j} \otimes O_{R,j}$, where the norm of each of the terms is bounded
	by $\|O_{AR}\|$. Here $d$ is the dimension of the elementary block, which is bounded by some constant.

	The unique fixed point of $\Phi_{r_A}^s$ is the identity operator, and $\tr[(\rho_s^{A_sR} - \rho_s^{A_s}\otimes \rho_s^R)\1\otimes O]=0$ for any operator $O$. Therefore,
	\begin{align}
		\delta&=\sum_{j=1}^{d^2}\tr[\rho_s^{A_sR} - \rho_s^{A_s} \otimes \rho_s^{R}] \nonumber\\
		&\sum_{k\neq 0} \lambda_k^{s-r_A}\tr [L_kO_{A_{r_A},j}]R_k \otimes O_{R,j}\nonumber\\
						       &\leq 2d^2 \|O_{AR} \|2^{-\nu(s-r_A)}.
	\end{align}
One can see that Eq.\ref{eq:correctability_bound} holds with a choice of constant $C=2d^2$.
\end{proof}

By invoking Theorem \ref{thm:decoupling}, we can easily show that the region $A$ is correctable.
\begin{corollary}
	For an RG-regular MERA code $\cC_s$, and for any simply connected region $A$, there exists a CPTP $\mathcal{R}$ acting on $\mathcal{H}_0$ such that
	\begin{equation}
		\|\mathcal{R}(\rho^{A^cR}) - \rho^{AA^cR} \|_1 \leq C2^{-\nu(s-\log_2|A|)/2},
	\end{equation}
	where $C$ is a numerical constant, and $A^c$ denotes the complement of $A$.
\end{corollary}
The proof simply follows by applying Theorem \ref{thm:decoupling} and then using the relation between the Bures distance and the trace distance (Eq.\ref{eq:Bures_Trace}).

Our findings support the conclusion of Ref. \cite{Pastawski2016}, in which it was suggested that low energies of the critical systems should have a certain error correction
property. In particular, our work provides a satisfying answer to the question: how does an error correcting code emerge in these systems? It arises from the fact that
the ground state can be well-approximated by a MERA state.

\subsection{Local correctability}
As was the case in Ref. \cite{Flammia2016}, the notion of local correctability plays an important role in our applications. We derive this for RG-regular MERA codes.

\begin{thm}\label{thm:local_correctability}
	Let $\cC_s$ be an RG-regular MERA code. Let $A$ be a simply connected region and let $B$ be a region shielding $A$ such that $AB$ is a set of sites that are distance $x$
	or less away from $A$ and $|AB|<2^s$.
	 $C$ is the complement of $AB$. Then there exists a recovery map $\mathcal{R}_{B}^{AB}: \mathcal{B}(\mathcal{H}_B) \to \mathcal{B}(\mathcal{H}_{AB})$ such that
	\begin{equation}
		\|\mathcal{R}_{B}^{AB}(\rho_0^{BCR}) - \rho_0^{ABCR} \|_1 \leq c\left(\frac{|A|}{x} \right)^{\nu/2} \label{eq:generalized_clustering}
	\end{equation}
	for all purified code states $\rho_0^{ABCR}$,
	where $c$ is a numerical constant.
\end{thm}
\begin{proof}
The proof is similar to that of Lemma \ref{lem:simply_connected}. First recall the following two identities.
\begin{equation}
\tr[\rho_0^{ACR}O_{ACR}]=\tr[\rho_s^{A_sC_sR}\Phi_0^s(O_{ACR})],
\end{equation}
\begin{equation}
\tr[\rho_0^{A}\otimes \rho_0^{CR}O_{ACR}] = \tr[\rho_s^{A_s}\otimes \rho_s^{C_sR} O_{ACR}],
\end{equation}
provided that the past causal cone of $A$ and $C$ do not overlap with each other all the way up to a scale $s.$
The first identity follows from the definition and the second identity follows from Lemma \ref{lem:clustering_identity_general}. Let us denote the left hand side of Eq. (\ref{eq:generalized_clustering}) as $\delta$. The two identities above imply that
\begin{align}
\delta &= \tr[(\rho_r^{A_rC_rR} - \rho_r^{A_r}\otimes \rho_r^{C_rR}) \Phi_0^r(O_{ACR})]\\
&=\tr[(\rho_r^{A_rC_rR} - \rho_r^{A_r}\otimes \rho_r^{C_rR})\Phi_{r_A}^{r}\Phi_0^{r_A}(O_{ACR})] \\
&= \tr[(\rho_r^{A_rC_rR} - \rho_r^{A_r}\otimes \rho_r^{C_rR})\Phi_{r_A}^{r}(O_{A'C'R})]
\end{align}
where $O_{A'C'R}=\Phi_{0}^{r_A}$ is an operator supported on $A'=A_{r_A}$, $C'=C_{r_A}$, and $R$. Here $r_A$ is chosen such that $A_{r_A}$ is contained in an elementary block and $r$ is chosen to be the scale after which the past causal cones of $A$ and $CR$ overlap with each other. This happens when $x$ is shrunk to a size of $O(1)$. Thus, it can be chosen to be $r=\log_2 x +O(1)$, where $O(1)$ is a non-universal constant of order unity.

Now consider the following operator Schmidt decomposition:
\begin{equation}
O_{A'C'R} = \sum_{j=1}^{d^2} O_{A', j}\otimes O_{C'R, j}.
\end{equation}
The action of $\Phi_{r_A}^{r}$ on this operator is of the following form:
\begin{equation}
\Phi_{r_A}^{r} = \sum_{j=1}^{d^2} \Phi_{r_A}^{r}(O_{A',j}) \otimes \Phi_{r_A}^{r}(O_{C'R,j}),
\end{equation}
due to the fact that the past causal cone of $A'$ and $C'$ do not overlap up to this point.

The unique fixed point of $\Phi_{r_A}^r$ is the identity operator, and $\tr[(\rho_r^{A_rC_rR} - \rho_r^{A_r}\otimes \rho_r^{C_rR}) \1\otimes O]=0$ for any operator $O$. The norm of the remaining $(d^2-1)$ terms are bounded by $\|O_{ACR}\|2^{-\nu(r-r_A)}$. By the aforementioned choice, i.e., $r=\log x + O(1)$ and $r_A= \log|A| + O(1)$, we conclude $r-r_A \geq \log_2(\frac{x}{|A|}) + O(1)$, thus yielding a bound on $\delta.$ By invoking Theorem \ref{thm:decoupling}, the theorem is proved.
\end{proof}

An important consequence of local correctability is that two distant correctable regions are jointly correctable. In the context of quantum error correction this property is called the union lemma. Indeed, suppose that regions $A_1$ and $A_2$ are both locally correctable on $A_1B_1$ and $A_2B_2$ up to error $\epsilon$ each. Then $A_1A_2$ is locally correctable on $A_1A_2B_1B_2$ up to error $2\epsilon$ if $A_1B_1\cap A_2B_2=\emptyset$ (see Lemma 11 in Ref. \cite{Flammia2016}).

\begin{figure*}
	\begin{center}
		\includegraphics[scale=0.30]{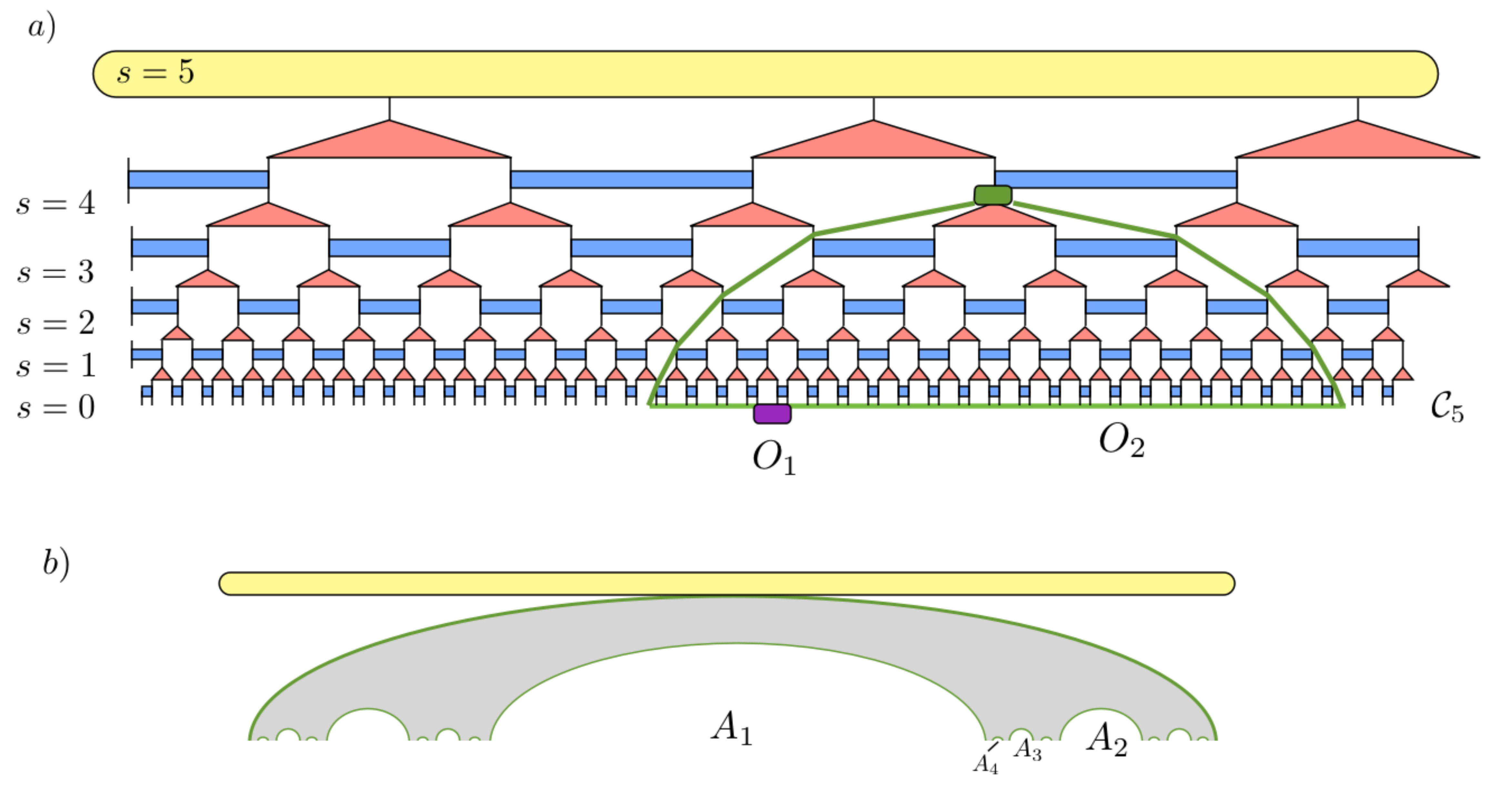}
	\caption{
      a) The setup described in Section \ref{sec:LR} where observables $O_1$ and $O_2$ act on $\cH_0$. Observable $O_2$ is a logical operator of $\cC_s$, in that it maps elements of $\cC_s$ to $\cC_s$. The Lieb-Robinson type bound of Eq. (\ref{eqn:LRB}) tells us how correlated the time evolution of $O_1$ is with respect to $O_2$.  b) The minimal correctable region of $\cC_s$ is also the minimal support of a logical operator, which corresponds to the distance of the error correcting code. We see that it takes a cantor-set type form, as already suggested in Ref. \cite{Pastawski2016a}. 	}\label{fig:MERA2}
			\end{center}
\end{figure*}

\section{Applications\label{section:applications}}
There are many implications of Theorem \ref{thm:local_correctability}. As was the case in Ref. \cite{Flammia2016} this forms the basis behind deriving fundamental tradeoff bounds for MERA codes. Furthermore, it also implies that two observables that are largely separated in scale compared to $1/\nu$ behave as if they are space-like separated operators. In particular, they obey a Lieb-Robinson type locality bound.

\subsection{Tradeoff bounds}

In this section we will derive bounds on the minimal support of a bulk logical operator on the boundary. In terms of quantum error correcting codes, this quantity corresponds to the \textit{distance} of the code.
For simplicity, we consider the  limit: $\nu \to \infty$. We do not expect this limit to be physical, because to our knowledge, no such theory is known at this point. However, it is the limit in which all of our statements become exact. In particular, we partially recover the so called `uberholography,' which was suggested recently by Pastawski and Preskill \cite{Pastawski2016a}. In this limit, all correlations outside the bulk lightcone vanish completely.

Consider an RG-regular MERA code $\cC_s$ of $n$ physical qubits. At this point, we do not restrict $|\cC_s|$ to being constant. From Theorem \ref{thm:local_correctability}, we know that  any state $\rho\in\cC_s$ can be recovered from $\rho_{A^c}$ by applying a channel on $AB$, provided that $x \geq |A|$, where $AB$ is a set of sites that are distance $x$ or less away from $A$. By choosing $x=|A|+O(1)$, we see that a subsystem $A$ of size less than $2^s/z$ can be locally corrected from such $B$, where $z= |AB|/|A|=3$.


Therefore, the logical information of $\cC_S$ can be recovered from these $N-2^s/z$ qubits. However, we can do better than this. By the so called union lemma\cite{Bravyi2008,Bravyi2010,Flammia2016}, two disconnected correctable regions $A_1$ and $A_2$ are jointly correctable if their local recovery maps have non-overlapping supports. This implies that $n/2^s+O(1)$ many regions of size $2^s/z$ are jointly correctable, implying that in fact only $n(1-\frac{1}{z})+O(1)$ many qubits are required to recover a code state.


It turns out that we can do even better. Let $\cR_{AB}$ be the map recovering erasure of region $A$. $\cR_{AB}$ takes as input the state $\rho_B$ and outputs $\rho_{AB}$: $\cR_{AB}(\rho_B)= \rho_{AB}$. $B$ is compose of a left and a right component: $B=B^LB^R$ such that $AB=B^LAB^R$. But if $B^L$ can be broken up into three regions $B^L=B_1^LA_1B_1^R$, where $B_1^{L,R}$ are the left and right parts of $B^L$, then we get: $\cR_{AB}\cR_{A_1B_1}(\rho_{B1})=\rho_{AB}$ (see Fig. \ref{fig:MERA2}b for an illustration). We can now iterate until we are left with $2^g$ regions of constant size. We now estimate what value $g$ can take. The smallest elements have size $O((\frac{z-1}{2z})^g|AB|)$, with $AB$ the original region. We want to know what fraction of $AB$ is left after the $g$ steps, or $2^g=|AB|^\alpha$. This yields
\begin{equation}
\alpha = \frac{\log(2)}{\log\left(\frac{2z}{z-1}\right)}.
\end{equation}

In terms of error correcting codes, we get that $|AB|=O(n/k)$, because $k=2^{\log_2(n)-s}$ and $|AB|=2^s$ by construction,   so that that distance (the smallest support of a logical operator) satisfies $d\leq C(n/k)^\alpha$, for some constant $C$, and $\alpha=\log(2)/\log(2z/(z-1))\approx 0.63$ for $z=3$. Note that our bound differs from Ref.\cite{Pastawski2016a}; there $\alpha \approx 0.78$, which yields a weaker bound. This is because our notion of local correctability is stronger than that of \cite{Pastawski2016a}; erasure of a simply connected region can be corrected if $x>|A|$ in our setup, but the setup of \cite{Pastawski2016a} requires $x>c|A|$, where $c\approx 2.414$. Indeed, the existence of the operators with small scaling dimensions in holography implies that the $\nu\to \infty$ limit cannot be an adequate description of such theories.



\subsection{Emergent lightcone}\label{sec:LR}
In this section, we establish a bound on how fast correlations between bulk local observables build up in time.
We will bound a commutator of the following form:
\begin{equation}
\bra{\rho_0}[O_1(t), O_2]\ket{\sigma_0},\label{eq:LR_emergent}
\end{equation}
where $O_1(t)$ is a local operator acting on $\mathcal{H}_0$, $O_2$ is a logical operator of $\cC_s$, and the time evolution is generated by a local Hamiltonian. We derive an upper bound, which remains  small provided that $|t|$ is  small compared to $2^{\nu s}$ up to some multiplicative constants. It is interesting to compare this bound to the well-known Lieb-Robinson bound \cite{Lieb1972}, which states that
\begin{equation}
\|[O_1(t),O_2]\|\leq c \| O_1 \| \| O_2 \|\exp(-\frac{L-v|t|}{\xi}), \label{eq:LR_original}
\end{equation}
where $L$ is the distance between the nontrivial support of $O_1$ and $O_2$, $v$ is the Lieb-Robinson velocity, and $\xi$ is a constant that depends on the underlying interaction graph. The main difference is that Eq. (\ref{eq:LR_original}) holds in the entire Hilbert space, while Eq. (\ref{eq:LR_emergent}) only holds in a low energy subspace, i.e., the code subspace $\cC_s$.

Obviously, a more refined bound would involve the size and the location of the supports of $O_1$ and $O_2$, but that is beyond the scope of this paper. Here we focus on a simpler setting, in which $O_1$ is assumed to be a local operator and $O_2$ to be an arbitrary logical operator in the code subspace.

Because the dynamics in the physical Hilbert space is assumed to be generated by a local Hamiltonian, observables under this time evolution obey Eq. (\ref{eq:LR_original}) with an appropriate choice of $v,c,$ and $\xi$. From this fact, we can derive the following bound.
\begin{thm}
For an RG-regular MERA code $\cC_s$, a local physical operator $O_1$ and a logical operator $O_2$ of $\cC_s$, we get
\begin{equation}
|\bra{\rho_0} [O_1(t),O_2] \ket{\sigma_0} |\leq c' (v|t| + \xi \nu s)^{\nu} 2^{-\nu s},\label{eqn:LRB}
\end{equation}
where $O_1(t)=e^{iHt} O_1 e^{-iHt}$, $c'$ is a constant, $\nu$ is the scaling dimension, $v$ is the Lieb-Robinson velocity of $H$, and $\xi$ is a numerical constant that depends on the interaction graph
of $H$.
\end{thm}
\begin{proof}
We consider the left hand side of Eq. (\ref{eqn:LRB})
\begin{equation}
\bra{\rho_0} [O_1(t),O_2] \ket{\sigma_0} = \bra{\rho_0}O_1(t)\ket{\sigma_0'} - \bra{\rho_0'}O_1(t)\ket{\sigma_0},\nonumber
\end{equation}
where $\ket{\sigma_0'} = O_2\ket{\sigma_0}$ and $\ket{\rho_0'}=O_2\ket{\rho_0}$ are states in $\cC_s$. From Eq. (\ref{eq:LR_original}) it follows that there exists an operator $O_1^l(t)$, supported on a set of sites with distance $l$ or less away from the support of $O_1$, such that \cite{Bravyi2006}
\begin{equation}
\|O_1(t) - O_1^l(t) \|\leq c \| O_1 \| \exp(\frac{l-v|t|}{\xi}).
\end{equation}

Therefore, both $|\bra{\rho_0} O_1(t) - O_1^l(t)\ket{\sigma_0'}|$ and $|\bra{\rho_0'} O_1(t) - O_1^l(t)\ket{\sigma_0}|$ are bounded by $c\| O_1\|  \exp(\frac{l-v|t|}{\xi})$. Further,  $\bra{\rho_0} O_1^l(t)\ket{\sigma_0'}=\bra{\rho_s} \Phi_0^s(O_1^l(t))\ket{\sigma_s'}$ and  $\bra{\rho_0'} O_1^l(t)\ket{\sigma_0'}=\bra{\rho_s'} \Phi_0^s(O_1^l(t))\ket{\sigma_s}$. Now, we can decompose the action of $\Phi_0^s$ into $\Phi_0^r$ and $\Phi_r^s$ so that $\Phi_0^r(O_1^l(t))$ is contained in an elementary block. Denoting this operator as $O'$,
\begin{align}
\bra{\rho_s} \Phi_r^s(O') \ket{\sigma_s'} - \bra{\rho_s'} \Phi_r^s(O') \ket{\sigma_s} =\delta \tr(O') +\epsilon,
\end{align}
where
\begin{equation}
	\delta = (\braket{\rho_s}{\sigma_s'} - \braket{\rho_s'}{\sigma_s})
\end{equation}
and
\begin{equation}
\epsilon = \sum_{k\neq 0} \lambda_k^{s-r} \tr[L_kO'] R_k.
\end{equation}
The first term vanishes because both $\braket{\rho_s}{\sigma_s'}$ and $\braket{\rho_s'}{\sigma_s}$ are equal to $\bra{\rho_0} O_2\ket{\sigma_0}$. The remaining term, $\epsilon$, is bounded by $2^{-\nu s}d^2\|O_1\|$. By choosing $\frac{l}{\xi} = \nu s + \frac{v|t|}{\xi}$, the bound is derived.

\end{proof}

It should be noted that the bound on $\bra{\rho_0}[O_1(t),O_2] \ket{\sigma_0}$ does not necessarily imply a bound on $\bra{\rho_0} [O_1,O_2(t)]\ket{\sigma_0}$. This is because the action of
the Hamiltonian may map a state in $\cC_s$ to a state outside of this subspace. However, this was to be expected, since we did not incorporate any relation between the code subspace and the
Hamiltonian. One solution is to consider the action of the commutator on states which are eigenstates of $H$. One physically reasonable choice would be the ground state of $H$. If the ground state
of $H$, $\ket{\psi}$, can be represented by a MERA such that the code subspace $\cC_s$ defined by this MERA is RG-regular, then
\begin{equation}
	\bra{\psi} [O_1(t),O_2]\ket{\psi} \leq c'(v|t|+\xi \nu s)^{\nu} 2^{-\nu s}
\end{equation}
because $\ket{\psi}$ by definition is in $\cC_s$. Furthermore, by construction $\bra{\psi}[O_1(t),O_2]\ket{\psi}=\bra{\psi}[O_1,O_2(t)]\ket{\psi}$.

At this point, we have only considered a Lieb-Robinson bound between an observable in the bulk, and another observable on the boundary inside the future light cone of the first. A more complete geometrical description of the Lieb-Robinson bounds for two observables anywhere in the bulk would be desirable. This is left for future work.

\section{Conclusion\label{section:conclusion}}
In this paper, we have outlined a mechanism by which certain toy models of holography can be derived from  generic properties of states at criticality. It is straightforward to see that the main findings of this paper, e.g., Theorem \ref{thm:local_correctability} and its implications follow in higher dimensions as well. This is because the derivation was based on a very general property of entanglement renormalization. This provides a partial explanation for the origin of these codes, assuming that the low energy states of conformal field theory can be well-described by MERA. It also explains how a causal dynamics can arise in these systems, despite the fact that the effective Hamiltonian in the bulk is not manifestly local.

However, many important issues remain. For one thing, it will be interesting to understand the entanglement wedge reconstruction\cite{Almheiri2014} in our framework. The approximate nature of our bound makes this analysis challenging.
It is also important to note that our bound is not strong enough to ensure locality at sub-AdS scale. This was to be expected because our bound does not incorporate the properties that are expected to be satisfied by conformal field theories with a semi-classical gravitational dual: that there is a gap in the scaling dimension of the operators.\cite{Heemskerk2009} For both of these issues, a framework that can organize operators in terms of their scaling dimensions is desirable. One possibility would be the operator algebra quantum error correction, as was suggested in Ref. \cite{Almheiri2014}, or an approximate version thereof. An analogous analysis would require a derivation of Theorem \ref{thm:decoupling} for general operator algebra, which may be of an independent interest.

In order to gain a more refined insight, it will be important to study families of circuits that are equipped with more refined set of structures. There are many interesting questions in this direction.
Would a random tensor network of Ref.\cite{Hayden2016} emerge from the random MERA network in Ref.\cite{Swingle2012a}? Can we import the constraints posed on the operators of the CFT into the language
of quantum error correction? Would the tradeoff bounds on quantum error correction lead to nontrivial constraints on gravity? These are left for future work.

\subsection*{Acknowledgments}
We especially thank Fernando Pastawski for helpful discussions. MJK was supported by the Villum foundation. IK would like to acknowledge the hospitality of Perimeter Institute and Simons center for geometry and physics, where part of this work was completed.

\bibliographystyle{JHEP}
\bibliography{bib2}

\section{Appendix}

Quantum channels are one of the most important tools in quantum information theory. They describe in the most general manner possible the evolution of a quantum system; meaning that they take as input a quantum state, and the output another quantum state. Quantum channels are completely positive and trace preserving operations. They have several useful representations, perhaps the most commonly used one is the Kraus form: $T(\cdot)=\sum_k E_k (\cdot)E_k^\dag$ for some set of ``Kraus operators" $\{E_k\}$ satisfying $\sum_k E_k^\dag E_k=\1$.  They can be represented as linear operators on Hilbert space as: $\hat{T}:\cH^2_a\rightarrow\cH^2_b$ for any general quantum channel $T:\cB(\cH_a)\rightarrow \cB(\cH_b)$, which in terms of Kraus operators reads:  $\hat{T}=\sum_k E_k\otimes \bar{E}_k$, where $\bar{E}$ is the complex conjugate of $E$.

In general, $\hat{T}$ is not a normal operator, meaning that it typically has Jordan blocks. In the special case when $\hat{T}$ is \textit{non-defective} (i.e. has no non-trivial Jordan blocks), then it has a spectral decomposition
\begin{equation}\label{eqn:channel2}
\hat{T}=\sum_j \lambda_j \ket{L_k}\bra{R_k}.
\end{equation}
If the input and output dimensions of the channel are the same then $\langle R_k | L_j\rangle =\delta_{jk}$ is a bi-orthonormal basis. The spectrum is bounded by one ($\lambda_j\leq 1$), and the channel always has at least one eigenvalue equal to one. If there is no other eigenvalue of magnitude one, then channel is called \textit{mixing}. In terms of CPTP maps, Eq. (\ref{eqn:channel}) reads
\begin{equation}
T(A)= \sum_k \lambda_k \tr[ A R_k]L_k.
\end{equation}

For more on quantum channels, see Refs. \cite{Nielsen2011,Wolf2012}


\end{document}